\documentclass[preprint,1p,times]{elsarticle}

\usepackage[T1]{fontenc}
\usepackage[utf8]{inputenc}
\usepackage{microtype}
\usepackage{amsmath,amssymb,amsthm,mathtools}
\usepackage{booktabs}
\usepackage{enumitem}
\usepackage[hidelinks]{hyperref}

\biboptions{numbers,sort&compress}

\newtheorem{theorem}{Theorem}
\newtheorem{lemma}[theorem]{Lemma}
\newtheorem{corollary}[theorem]{Corollary}
\theoremstyle{definition}

\newcommand{\Ccal}{\mathcal{C}}
\newcommand{\Ucal}{\mathcal{U}}
\newcommand{\Scal}{\mathcal{S}}
\newcommand{\Pcal}{\mathcal{P}}
\newcommand{\Z}{\mathbb{Z}}
\newcommand{\Q}{\mathbb{Q}}
\newcommand{\val}{\operatorname{val}}
\newcommand{\vc}{\operatorname{vc}}
\newcommand{\HomNC}{\textnormal{\textsc{HomNC}}}

\hypersetup{
  pdftitle={Homogeneous Network Caching is Fixed-Parameter Tractable Parameterized by the Number of Caches},
  pdfauthor={Jozsef Pinter and Regina Stangl}
}

\begin{document}

\begin{frontmatter}

\title{Homogeneous Network Caching is Fixed-Parameter Tractable Parameterized by the Number of Caches}

\author[inst1,inst2]{J\'ozsef Pint\'er\corref{cor1}}
\ead{pinterj@edu.bme.hu}
\author[inst1]{Regina Stangl}
\ead{stangl.regina@edu.bme.hu}
\cortext[cor1]{Corresponding author.}

\affiliation[inst1]{organization={Department of Stochastics, Institute of Mathematics, Budapest University of Technology and Economics},
            city={Budapest},
            country={Hungary}}

\affiliation[inst2]{organization={HUN-REN--BME Stochastics Research Group},
            city={Budapest},
            country={Hungary}}

\begin{abstract}
Network caching asks how to place contents in distributed caches so that future requests are served close to their users. Ganian, Mc Inerney and Tsigkari recently initiated the parameterized-complexity study of the problem and, for the homogeneous unit-size variant (\HomNC{}), isolated an unresolved family of six parameterizations: by the number of caches $C$, the number of users $U$, $U+K$, $C+U$, $C+\lambda$, and the vertex-cover number $\vc(G)$, where $K$ is the maximum cache capacity and $\lambda$ is the maximum number of contents requested with nonzero probability by any user. Their interreducibility theorem showed that these six cases stand or fall together under parameterized reductions, and they conjectured the family to be W[1]-hard. We resolve this conjecture in the opposite direction. We prove that \HomNC{} is fixed-parameter tractable parameterized by $C$ alone, and therefore fixed-parameter tractable for all six parameterizations. Our algorithm is based on an exact $n$-fold integer programming formulation that reveals a nontrivial block structure in homogeneous network caching, with the repeated part depending only on $C$. Standard algorithms for $n$-fold integer programming then yield a running time of the form $f(C)\lvert I\rvert^{O(1)}$.
\end{abstract}


\end{frontmatter}

\section{Introduction}

Caching is a standard way to improve the performance of networked
systems and reduce congestion: reusable content is stored closer to
users, either reactively or through proactive placement. Related
placement questions arise in content-delivery networks, wireless
helper systems, and information-centric networking
\cite{paschos2020cache,shanmugam2013femtocaching,zhang2013caching}.
In such applications, cache placement can reduce latency, backhaul
traffic, and peak load, but it must respect finite storage capacities
and potentially complex user--cache access constraints. Proactive caching has also been
proposed as a way to exploit demand prediction in wireless and
edge-like networks \cite{bastug2014living}; even the uncoded
helper-cache placement problem is NP-hard
\cite{shanmugam2013femtocaching}.

This note studies the homogeneous network-caching problem, denoted \HomNC{}, where all contents have unit size. Homogeneity is natural when files are split into equal-size chunks, or when the placement decision is made at the level of standardized cache objects. The model still retains the central combinatorial features of cache placement. Caches may have arbitrary capacities; users may access arbitrary subsets of caches; and users may have different weights and different request distributions. The objective has an OR semantics: a user obtains value from a content if at least one adjacent cache stores it, and additional adjacent copies of the same content do not increase that user's value.

Ganian, Mc Inerney and Tsigkari recently gave a systematic parameterized-complexity study of network caching \cite{ganian2025parameterized}. They considered the homogeneous unit-size problem together with heterogeneous variants and proved a broad collection of upper and lower bounds. For \HomNC{}, the unresolved homogeneous family consisted of the following six parameterizations:
\[
  C,
  \qquad U,
  \qquad U+K,
  \qquad C+U,
  \qquad C+\lambda,
  \qquad \vc(G).
\]
Here $C$ is the number of caches, $U$ is the number of users, $K$ is the maximum cache capacity, $\lambda$ is the maximum number of contents requested with nonzero probability by a user, and $\vc(G)$ is the vertex-cover number of the caching graph. All six of these appear as open cases in their main homogeneous landscape. Their interreducibility theorem shows that the six parameterized problems are pairwise connected by parameterized reductions \cite[Theorem~14]{ganian2025parameterized}. Ganian, Mc Inerney and Tsigkari conjectured that the family is W[1]-hard.

We resolve the open family on the positive side.

\begin{theorem}[Main theorem]\label{thm:main}
\HomNC{} can be solved in time $f(C)|I|^{O(1)}$ for a computable function $f$, where $C$ is the number of caches and $|I|$ is the binary encoding length of the input. Thus \HomNC{} is fixed-parameter tractable parameterized by the number of caches.
\end{theorem}

The proof is based on a bounded-interface view of the problem. If $C$ is fixed, then each content has only $2^C$ possible placement patterns, namely the subsets of caches in which it is stored. Once each content has chosen one of these patterns, different contents interact only through the $C$ cache-capacity constraints. This is exactly the structure of an $n$-fold integer program: one bounded brick per content and a bounded number of global coupling rows. 

It is useful to spell out why the earlier methods do not already imply Theorem~\ref{thm:main}. Ganian, Mc Inerney and Tsigkari's capacity-vector dynamic program stores, for every vector of remaining cache capacities, the best value obtained after processing a prefix of the contents \cite[Theorem~3]{ganian2025parameterized}. With $C$ caches this state space contains
\[
  \prod_{i=1}^{C}(\kappa(c_i)+1)
\]
entries, corresponding to $(\kappa(c_i)+1)$ remaining-capacity values per cache. In \HomNC{} one may truncate each capacity at $|\Scal|$, since allocations are sets and no cache can usefully store more than the whole catalog, but the bound then becomes $(|\Scal|+1)^C$: this is XP in $C$, not FPT in $C$. 

The interreducibility theorem shows that the open cases have the same parameterized status, but it does not decide that status. Our formulation is different at the technical point where this matters: capacities are not coordinates of a dynamic-programming state, but right-hand sides and bounds in an integer program whose repeated matrices depend only on $C$.

By combining Theorem~\ref{thm:main} with the interreducibility result of Ganian, Mc Inerney, and Tsigkari~\cite[Theorem~14]{ganian2025parameterized}, we obtain the following corollary. We also prove it directly in Section~\ref{sec:consequences} to keep the paper self-contained.

\begin{corollary}\label{cor:intro}
\HomNC{} is fixed-parameter tractable parameterized by each of
\[
  C,
  \qquad U,
  \qquad U+K,
  \qquad C+U,
  \qquad C+\lambda,
  \qquad \vc(G).
\]
\end{corollary}

The rest of the note is organized as follows. Section~\ref{sec:definitions} fixes the \HomNC{} model and recalls the $n$-fold IP framework. Section~\ref{sec:main} proves Theorem~\ref{thm:main}. Section~\ref{sec:consequences} presents the resulting consequences and their direct proofs. Section~\ref{sec:conclusion} concludes.

\section{Definitions}\label{sec:definitions}

A \emph{caching graph} is a bipartite graph
\[
  G=(\Ccal\cup\Ucal,E),
\]
where $\Ccal$ is the set of caches and $\Ucal$ is the set of users. For a vertex $v$, let $N(v)$ be its open neighborhood. Thus $N(u)\subseteq\Ccal$ is the set of caches accessible to user $u$, while $N(c)\subseteq\Ucal$ is the set of users served by cache $c$.

A \HomNC{} instance consists of $G$, a finite catalog $\Scal$ of contents, cache capacities $\kappa(c)\in\Z_{\ge 0}$, request probabilities $p_{u,s}\in\Q_{\ge 0}$, user weights $w(u)\in\Q_{\ge 0}$, and a target value $\ell\in\Q_{\ge 0}$. If $\Scal=\emptyset$, the optimum is $0$ and the instance is decided immediately. Otherwise, for every user $u$, the values $(p_{u,s})_{s\in\Scal}$ form a probability distribution over $\Scal$. All contents have unit size.

An allocation assigns a set $A(c)\subseteq\Scal$ to each cache $c$ such that
\[
  |A(c)|\le \kappa(c).
\]
The contents visible to user $u$ are
\[
  R_A(u)=\bigcup_{c\in N(u)} A(c),
\]
and the cache-hit value is
\begin{equation}\label{eq:value}
  \val(A)=\sum_{u\in\Ucal} w(u)\sum_{s\in R_A(u)}p_{u,s}.
\end{equation}
The decision problem asks whether there exists an allocation $A$ with $\val(A)\ge \ell$.

We use the parameters
\[
  C=|\Ccal|,
  \qquad
  U=|\Ucal|,
  \qquad
  K=\max_{c\in\Ccal}\kappa(c),
\]
where $K=0$ if $\Ccal=\emptyset$, and
\[
  \lambda=
  \max_{u\in\Ucal}|\{s\in\Scal:p_{u,s}>0\}|,
\]
where $\lambda=0$ if $\Ucal=\emptyset$. The vertex-cover number of $G$ is denoted by $\vc(G)$.

All numbers are encoded in binary. Put $\alpha_{u,s}=w(u)p_{u,s}$. To use an integer objective, let $D$ be the product of the denominators of all nonzero rational numbers $\alpha_{u,s}$ and of $\ell$, after writing them in lowest terms. The bit-length of $D$ is at most the sum of these denominator bit-lengths, and is therefore polynomial in the input length. Using the least common multiple instead would only decrease $D$. Multiplication by $D$ preserves the answer:
\[
  \val(A)\ge \ell
  \quad\Longleftrightarrow\quad
  D\val(A)\ge D\ell.
\]
Let $T=D\ell$.

We next recall the only external algorithmic tool. For matrices $A\in\Z^{r\times t}$ and $B\in\Z^{q\times t}$, their $n$-fold product is
\[
[A,B]^{(n)}=
\begin{pmatrix}
A&A&\cdots&A\\
B&0&\cdots&0\\
0&B&\cdots&0\\
\vdots&\vdots&\ddots&\vdots\\
0&0&\cdots&B
\end{pmatrix}.
\]
The first $r$ rows are global rows, and the remaining $nq$ rows are local rows, one group for each brick.

\begin{theorem}[$n$-fold IP algorithms]\label{thm:nfold}
Let $A\in\Z^{r\times t}$ and $B\in\Z^{q\times t}$ be fixed block matrices, and let $\Gamma=\max\{\|A\|_\infty,\|B\|_\infty,1\}$. Given $n$, binary-encoded finite integer lower and upper bounds, a binary-encoded integer right-hand side, and a binary-encoded integer linear objective whose coefficients may vary between bricks, integer linear optimization over
\[
  \max\ w^\top x
  \quad\text{s.t.}\quad
  [A,B]^{(n)}x=b,
  \qquad l\le x\le u,
  \qquad x\in\Z^{nt},\ w\in\Z^{nt},
\]
is fixed-parameter tractable parameterized by $r,q,t$ and $\Gamma$. Equivalently, if the block matrices depend only on a parameter $k$, then the program can be solved in time $g(k)L^{O(1)}$, where $L$ is the binary encoding length of all numerical input, including $n$, $b$, $l$, $u$, and the objective coefficients.
\end{theorem}

Theorem~\ref{thm:nfold} follows from the algorithmic theory of $n$-fold integer programming \cite{de2008n,hemmecke2013n,jansen2020near}. We use it only as a black box. What matters here is that the number of bricks may be large, and the right-hand sides, bounds, and objective coefficients may be binary encoded and may vary between bricks, as long as the repeated block shape is parameter-bounded. For example, the near-linear algorithm of Jansen, Lassota and Rohwedder gives, for finite bounds, a running time of the form
\[
  (rq\Gamma)^{O(r^2q+q^2)} L^2\, nt\log^{O(1)}(nt),
\]
using their notation. Applied to our formulation this gives an explicit computable dependence on $C$.

\section{The main theorem: FPT by the number of caches}\label{sec:main}

If $\Ccal=\emptyset$ or $\Scal=\emptyset$, the optimum is $0$, so the instance is immediate. Therefore, we assume $C\ge 1$ and $\Scal\ne\emptyset$. Enumerate the caches as
\[
  \Ccal=\{c_1,\ldots,c_C\}.
\]
A placement pattern is a subset $P\subseteq [C]$. It means that a content is stored exactly in the caches $c_i$ with $i\in P$; the empty pattern is allowed and means that the content is not stored anywhere. Let $\Pcal=2^{[C]}$.

For each content $s\in\Scal$ and pattern $P\in\Pcal$, define
\begin{equation}\label{eq:patternvalue}
  v_s(P)=
  \sum_{\substack{u\in\Ucal:\ N(u)\cap\{c_i:i\in P\}\ne\emptyset}}
  \alpha_{u,s},
  \qquad
  V_s(P)=Dv_s(P)\in\Z.
\end{equation}
Thus $v_s(P)$ is exactly the contribution of content $s$ if it is placed according to $P$: user $u$ contributes precisely when at least one accessible cache stores $s$.

For every content $s$ and pattern $P$, introduce a variable $x^s_P\in\Z_{\ge0}$, with intended meaning $x^s_P=1$ if $s$ chooses pattern $P$. For every content $s$ and cache index $i\in[C]$, introduce a slack variable $h^s_i\in\Z_{\ge0}$. The integer program is
\begin{align}
  \sum_{P\in\Pcal}x^s_P &= 1 &&\text{for every }s\in\Scal, \label{eq:local}\\
  \sum_{s\in\Scal}\left(\sum_{\substack{P\in\Pcal:\ i\in P}}x^s_P+h^s_i\right) &= \kappa(c_i) &&\text{for every }i\in[C], \label{eq:global}\\
  \max\quad \sum_{s\in\Scal}\sum_{P\in\Pcal}V_s(P)x^s_P. \label{eq:objective}
\end{align}
We impose the finite bounds
\[
  0\le x^s_P\le 1,
  \qquad
  0\le h^s_i\le \kappa(c_i).
\]
The variables $h^s_i$ merely turn capacity inequalities into equalities, which is the standard form for $n$-fold IP. Their upper bounds do not require capacities to be small; capacities are binary-encoded numerical bounds, and any residual capacity of cache $c_i$ is at most $\kappa(c_i)$.

Let one brick be
\[
  z^s=((x^s_P)_{P\in\Pcal},(h^s_i)_{i\in[C]})
  \in\Z^{2^C+C}_{\ge0}.
\]
Let $B_C$ be the row vector with coefficient $1$ on every pattern variable and coefficient $0$ on every slack variable. Let $A_C$ have one row for each cache $i\in[C]$; in row $i$, the coefficient of pattern column $P$ is $1$ if $i\in P$ and $0$ otherwise, and the coefficient of slack column $h_j$ is $1$ if $i=j$ and $0$ otherwise.

\begin{lemma}\label{lem:nfold}
The constraints \eqref{eq:local}--\eqref{eq:global} form an $n$-fold integer program with $n=|\Scal|$ bricks, $C$ global rows, one local row per brick, $2^C+C$ variables per brick, and coefficient bound $1$. The block matrices depend only on $C$.
\end{lemma}

\begin{proof}
Each content contributes one brick $z^s$. The local matrix $B_C$ enforces \eqref{eq:local}. The global matrix $A_C$ records how a chosen pattern contributes one unit of load to each cache and how the slack columns fill unused capacity. The same matrices $A_C$ and $B_C$ are used for every content, all entries are in $\{0,1\}$, and their dimensions are determined solely by $C$.
\end{proof}

\begin{lemma}\label{lem:exact}
The integer program \eqref{eq:local}--\eqref{eq:objective} has an integer feasible solution of objective value at least $T$ if and only if the original \HomNC{} instance has an allocation of value at least $\ell$.
\end{lemma}

\begin{proof}
Let $A$ be a feasible allocation. For every content $s$, define
\[
  P_s=\{i\in[C]:s\in A(c_i)\}.
\]
Set $x^s_{P_s}=1$ and set all other pattern variables for $s$ to zero. Then \eqref{eq:local} holds. For cache $c_i$, let $m_i=|\{s:i\in P_s\}|$. Feasibility gives $m_i\le\kappa(c_i)$. Choose any content $s_0\in\Scal$, set $h^{s_0}_i=\kappa(c_i)-m_i$ for every $i$, and set all other slack variables to zero. Since $0\le \kappa(c_i)-m_i\le \kappa(c_i)$, equation \eqref{eq:global} holds and all slack bounds are respected. Finally, by \eqref{eq:patternvalue}, the objective value is
\[
  \sum_{s\in\Scal}V_s(P_s)=D\val(A).
\]
Thus an allocation of value at least $\ell$ yields an IP solution of value at least $T$.

Conversely, let an integer feasible solution of the IP be given. By \eqref{eq:local}, nonnegativity, and the bounds $x^s_P\le1$, each content $s$ selects a unique pattern $P_s$ with $x^s_{P_s}=1$. Define
\[
  A(c_i)=\{s\in\Scal:i\in P_s\}.
\]
For every cache $c_i$, equation \eqref{eq:global} gives
\[
  |A(c_i)|+\sum_{s\in\Scal}h^s_i=\kappa(c_i).
\]
Since the slack variables are nonnegative, $|A(c_i)|\le\kappa(c_i)$, so $A$ is feasible. The users seeing content $s$ are exactly those whose neighborhood intersects the caches indexed by $P_s$. Hence the scaled value of $A$ is $\sum_s V_s(P_s)$, the IP objective. An IP solution of value at least $T=D\ell$ therefore gives an allocation of value at least $\ell$.
\end{proof}

\begin{proof}[Proof of Theorem~\ref{thm:main}]
By Lemma~\ref{lem:nfold}, the optimization step is an $n$-fold IP with
\[
  r=C,
  \qquad q=1,
  \qquad t=2^C+C,
  \qquad \Gamma=1.
\]
Theorem~\ref{thm:nfold} solves it in time $f(C)L^{O(1)}$, where $L$ is the binary encoding length of the constructed IP. The transformation above has polynomial bit complexity after denominator clearing, so $L$ is polynomial in the original input length for fixed $C$.

It remains only to compute the coefficients $V_s(P)$. For each user $u$, store the $C$-bit mask of $N(u)$. Then $u$ contributes to $v_s(P)$ exactly when this mask intersects the mask of $P$. With a dense request matrix all coefficients can be computed in $O(|\Scal|2^C|\Ucal|)$ arithmetic operations, with polynomial bit complexity; with a sparse representation one iterates only over the listed nonzero requests and treats missing entries as zero, which is still $f(C)\operatorname{poly}(|I|)$. By Lemma~\ref{lem:exact}, deciding whether the optimum IP value is at least $T$ decides the original instance. The total running time is $f(C)|I|^{O(1)}$.
\end{proof}

\section{Consequences for the other homogeneous parameterizations}\label{sec:consequences}

Theorem~\ref{thm:main} proves fixed-parameter tractability for the parameter $C$. Ganian, Mc Inerney and Tsigkari's interreducibility theorem also transfers this result to all six parameterizations listed in Corollary~\ref{cor:intro} \cite[Theorem~14]{ganian2025parameterized}. We give a direct proof as well, because it is short and clarifies why no bounded-capacity assumption is needed.

\begin{lemma}\label{lem:cachetwins}
In a \HomNC{} instance, two caches with the same user-neighborhood can be replaced by one cache with that neighborhood and capacity equal to the sum of the two capacities. The optimum value is preserved.
\end{lemma}

\begin{proof}
Let $c$ and $c'$ be caches with $N(c)=N(c')$. Any allocation in the original instance maps to the merged cache by storing $A(c)\cup A(c')$. This union has size at most $|A(c)|+|A(c')|\le \kappa(c)+\kappa(c')$, and duplicate copies are irrelevant because contents have unit size and user value has OR semantics. Conversely, any set of at most $\kappa(c)+\kappa(c')$ contents stored in the merged cache can be partitioned into two sets of sizes at most $\kappa(c)$ and $\kappa(c')$ and assigned to $c$ and $c'$. In both directions, exactly the same users see exactly the same contents. Hence the optimum is unchanged.
\end{proof}

\begin{proof}[Proof of Corollary~\ref{cor:intro}]
The case $C$ is Theorem~\ref{thm:main}. The cases $C+U$ and $C+\lambda$ are immediate because they bound $C$.

If $U$ is bounded, Lemma~\ref{lem:cachetwins} leaves at most $2^U$ cache-neighborhood types. Hence the number of caches becomes bounded by a function of $U$, and Theorem~\ref{thm:main} applies. This also proves the case $U+K$, since $U+K$ bounds $U$.

It remains to discuss the vertex-cover parameter. Let $p=\vc(G)$ and compute a vertex cover $X$ with $|X|\le 2p$ by taking all endpoints of a maximal matching. Every cache outside $X$ has all its neighbors in $X\cap\Ucal$: otherwise an edge from such a cache to a user outside $X$ would be uncovered. Hence outside caches have at most $2^{|X\cap\Ucal|}$ distinct user-neighborhoods. After applying Lemma~\ref{lem:cachetwins} to caches outside $X$, the number of remaining caches is at most
\[
  |X\cap\Ccal|+2^{|X\cap\Ucal|}
  \le
  2p+2^{2p}.
\]
This is a function of $\vc(G)$ only, so Theorem~\ref{thm:main} applies again. Equivalently, all six cases follow from Theorem~\ref{thm:main} and the interreducibility theorem of Ganian, Mc Inerney and Tsigkari.
\end{proof}

Thus the homogeneous family conjectured to be W[1]-hard is in fact fixed-parameter tractable. 

\section{Conclusion}\label{sec:conclusion}

Homogeneous network caching is fixed-parameter tractable parameterized by the number of caches. The proof identifies the bounded interface of the problem: every content chooses one of $2^C$ placement patterns, and the only global interaction is through the $C$ cache-capacity constraints. This yields an exact $n$-fold integer program and avoids the XP dependence inherent in capacity-indexed dynamic programs.

Together with the interreducibility theorem of Ganian, Mc Inerney and Tsigkari, this resolves the remaining homogeneous parameterized cases of network caching on the FPT side. The formulation also gives a compact exact model for instances with few caches: precompute the pattern values and solve the resulting capacity-coupled integer program. The proof uses homogeneity essentially, since non-unit content sizes would make capacity contributions depend on the content brick and would destroy the identical-block formulation without additional structure. A natural next question is whether the special $0$--$1$ structure of the block matrices can lead to sharper dependence on $C$ than what follows from general $n$-fold IP algorithms.

\section*{Funding}
J\'ozsef Pint\'er is funded by Project No. KDP-IKT-2023-900-I1-00000957/0000003 with support provided by the Ministry of Culture and Innovation of Hungary from the National Research, Development and Innovation Fund, financed under the KDP-2023 funding scheme. 

\section*{Declaration of competing interest}
The authors declare that they have no competing interests.

\section*{Data availability}
No datasets were generated or analyzed during the current study.

\section*{Declaration of generative AI and AI-assisted technologies in the manuscript preparation process}
During the preparation of this work, the authors used OpenAI's ChatGPT for brainstorming, organization, \LaTeX{} drafting, and editorial assistance. The authors reviewed and edited the output and take full responsibility for the content of the manuscript.

\section*{Acknowledgements}
The authors thank Robert Ganian, Fionn Mc Inerney and Dimitra Tsigkari for initiating the parameterized study of network caching and for formulating the homogeneous open cases resolved here.

\bibliographystyle{elsarticle-num}
\bibliography{references}

\end{document}